\documentclass[british,a4paper,10pt]{article} 

\usepackage{amssymb}
\usepackage{amsmath}
\usepackage{subfigure}
\usepackage{slashed}
\usepackage{amsmath}
\usepackage[normalem]{ulem}
\usepackage{graphicx,psfrag}
\usepackage{mathrsfs}
\usepackage{amsfonts}
\usepackage{multirow}
\usepackage{hyperref}
\usepackage{placeins}
\usepackage{amsthm}
\usepackage{latexsym}
\usepackage[dvips]{epsfig}
\usepackage{enumitem}  
\usepackage{soul}
\usepackage{graphicx} 
\usepackage{amsthm}
\usepackage{latexsym}   
\usepackage[dvips]{epsfig}
\usepackage[utf8]{inputenc}
\usepackage[english]{babel}
\usepackage{multicol}
\usepackage{color}
\usepackage{comment}
\usepackage[dvipsnames]{xcolor}
\usepackage{wasysym}
\usepackage{mathrsfs}
\usepackage{eufrak}
\usepackage{bm}
\usepackage{authblk}
\usepackage{slashed}
\usepackage{yhmath} 
\usepackage{stmaryrd}

\usepackage{lipsum}
\usepackage{tensor}

\theoremstyle{plain}
\newtheorem{proposition}{Proposition}

\newtheorem{theorem}{Theorem}
\newtheorem{assumption}{Assumption}

\newtheorem{remark}{Remark}

\def\bmg{{\bm g}}

\def\bmh{{\bm h}}

\begin{document}

\title{\textbf{Twistor initial data characterisation of pp-waves}}
\author[1]{Edgar Gasper\'in
  \footnote{E-mail address:{\tt edgar.gasperin@tecnico.ulisboa.pt}}}
\author[,2]{Francisco Pais
   \footnote{E-mail address:{\tt franciscopais@tecnico.ulisboa.pt}}}
  \affil[1,2] {CENTRA, Departamento de F\'isica, Instituto Superior
    T\'ecnico IST, Universidade de Lisboa UL, Avenida Rovisco Pais 1,
    1049 Lisboa, Portugal}


\maketitle
  
\begin{abstract}
  This note gives a concise derivation of a twistor-initial-data
  characterisation of pp-wave spacetimes in vacuum. The construction
  is based on a similar calculation for the Minkowski spacetime in
  [Class. Quantum Grav. 28 075010]. The key difference is that for the Minkowski
  spacetime a necessary condition is that
  $\nabla_{A}{}^{A'}\bar{\kappa}_{A'} \neq 0$.  In this note it is
  shown that if $\nabla_{A}{}^{A'}\bar{\kappa}_{A'}=0$ then the development
  is a pp-wave spacetime. Furthermore, it is shown that such
  condition propagates off the initial hypersurface, which, in turn,
  gives a \emph{twistor initial data characterisation of pp-waves}.
\end{abstract}



\section*{}

The existence of symmetries encoded through Killing objects (spinors,
vectors, tensors, Killing-Yano tensors) in a spacetime is a strong
constraint that can be exploited for obtaining geometric
characterisations of spacetimes of physical interest. In the context
of the Cauchy problem in General Relativity, whether an initial data
for the Einstein field equations will develop into a spacetime
admitting one of these Killing objects can be determined through
Killing initial data equations. The prototypical examples are the
Killing vector and Killing spinor initial data equations of
\cite{BeiChr97b} and \cite{GarVal08c}, respectively. These initial
data conditions can, in turn, be exploited to obtain initial data
characterisations of particular spacetimes. Examples of this
construction include the characterisation of the Kerr spacetime in
\cite{GarVal08a}, exploiting the (valence-2) Killing spinor initial
data equations, and the simpler characterisation of the Minkowski
spacetime through twistors (valence-1 Killing spinor) in
\cite{BaeVal11a}.  In this note, we show that, by augmenting the
conditions imposed for the standard twistor initial data of
\cite{GarVal08c}, one can derive an initial data characterisation of
pp-wave spacetimes, an approach which can be seen as the spinorial
analogue of the pp-wave initial data characterisation via conformal
Killing vectors of \cite{Gar22}. As pointed out in \cite{Gar22}, an
initial data characterisation of pp-waves can be obtained by employing
Killing spinors, and has been carried out in length in \cite{MurSha20}
---albeit in a different language and with a different scope. In this
note it is shown how a similar characterisation can be obtained 
concisely through a simple modification of twistor initial data
equations of \cite{GarVal08c, BaeVal11a}.

\medskip

Let $(\mathcal{M},\bmg)$ be a 4-dimensional manifold equipped with a Lorentzian metric $\bmg$
of signature $(+,-,-,-)$ and a spinor structure.
Any non-trivial spinor $\kappa_A$ satisfying
\begin{align}
\nabla_{A'(A}\kappa_{B)}=0,
\end{align}
will be referred to as a valence-1 Killing spinor or simply as a \emph{twistor}.
The integrability condition for the last equation is $\Psi_{ABCD} \, \kappa^D=0$
which restricts the spacetime to be of Petrov type N or O.

\medskip

From this point onward, unless otherwise stated, it will be assumed
that the vacuum Einstein field equations (without cosmological
constant) are satisfied. In other words, it will be assumed that
$(\mathcal{M},\bmg)$ is Ricci-flat $R_{ab}=0$, which, in spinorial
Newman-Penrose notation, is encoded through the vanishing of
trace-free Ricci spinor and the Ricci scalar
($\Phi_{AA'BB'}=\Lambda=0$).  One can give a spacetime
characterisation of the Minkowski spacetime through the existence of a
twistor as follows:
 \begin{proposition}[B\"ackdahl \& Valiente-Kroon ]\label{prop1_mink}
   If $\kappa_A$ is a twistor in an asymptotically flat spacetime
   $(\mathcal{M},\bmg)$ and $\eta_A :=\nabla_A{}^{A'}\bar{\kappa}_{A'}
   \neq 0$ at some point $p\in \mathcal{M}$, then the spacetime is the
   Minkowski spacetime.
 \end{proposition}
 \begin{proof}
 In short, the proof of this proposition is based on the observation
 that, if $\kappa_A$ is a twistor for which $\eta_{A}\neq 0$, then one
 has $\Psi_{ABCD}\kappa^A=\Psi_{ABCD}\eta^A=0$. These conditions imply
 that one can construct an adapted spin dyad $\{\kappa, \eta\}$ such
 that $\Psi_{ABCD}=0$ ---see \cite{BaeVal11a} for the detailed proof.
 \end{proof}

 \medskip

 A plane-fronted wave with parallel rays, or pp-wave for short, is a
 solution to the Einstein field equations in vacuum characterised by
 the existence of a null covariantly constant vector $k^a$ ---see
 \cite{Gar22}. This in turn implies that there exist a local
 coordinate system $(u,r,x^i)$ with $i =1,2$ for which the metric
 reads
 
    \begin{align*}
g_{H}=2\mathcal{H}(u,x^i)\mbox{d}u^2+2\mbox{d}u\mbox{d}r  - \delta_{ij}\mbox{d}x^i\mbox{d}x^j
    \end{align*}
    where $\bm\delta$ is the 2-dimensional Euclidean metric, and
    \begin{equation} \label{HppwaveHarmonic}
      \delta^{ij}\partial_i\partial_j \mathcal{H}=0.
    \end{equation}
    It is well-known that not every pp-wave spacetime is globally
    hyperbolic ---see \cite{Pen65b}. However, global hyperbolicity of
    this class of spacetimes strongly depends on the behaviour of
    $\mathcal{H}$ at spatial infinity. In \cite{FloSan04} the
    conditions for a pp-wave to be strongly hyperbolic have been
    established.

    \medskip

    The condition $\eta_{A} = 0$ (explicitly excluded in Proposition
    \ref{prop1_mink}) is key for the characterisation of pp-wave
    spacetimes.  To see this and to set up the notation, let
    $H_{A'AB}:=2\nabla_{A'(A}\kappa_{B)}$ and
    $\bar{\eta}_{A'}:=\nabla_{A'}^Q\kappa_Q$. Then, the irreducible
    decomposition of $\nabla_{AA'}\kappa_{B}$ reads
    
 \begin{align}
   2\nabla_{A'A}\kappa_B = H_{A'AB}+ \epsilon_{AB}\bar{\eta}_{A'}.
 \end{align}
 If $\kappa_A$ is a twistor then $H_{A'AB}=0$ and, if in addition,
 $\bar{\eta}_{A'}=0$, then $\nabla_{AA'}\kappa_B=0$. Consequently,
 $\kappa^A \bar{\kappa}^{A'}$ is a covariantly constant vector. Hence,
 the condition $ \nabla_{AQ'} \kappa^{Q'}=0$ ensures that
 Proposition \ref{prop1_mink} does not apply and that
 $\kappa^{A} \bar{\kappa}^{A'}$ is covariantly constant. One then
 concludes that the  spacetime is a pp-wave. 
 This discussion is summarised in the following:
 
  \begin{proposition}\label{prop1_ppwave}
   If $\kappa_A$ is a twistor for which $\bar{\eta}_{A'}
   :=\nabla_{A'}{}^{A}\kappa_{A} = 0$, then $(\mathcal{M},\bmg)$ is a
   pp-wave spacetime for some function $\mathcal{H}$ satisfying
   \eqref{HppwaveHarmonic}.
 \end{proposition}

  Proposition \eqref{prop1_ppwave} amounts to a spacetime
  characterisation; one can, however, obtain a characterisation at the
  level of initial data by slightly modifying the \emph{twistor
  initial data equations} of \cite{GarVal08c}. Before doing so, we
  first give a brief discussion of the derivation of the twistor
  initial data equations.  The twistor initial data result of
  \cite{GarVal08c} is based on the following identities which hold in
  Ricci-flat spacetimes:
 \begin{align}
   \square H_{A'AB} & = 2\nabla_{A'(A}\square \kappa_{B)} +
   2\Psi_{AB}{}^{PQ}H_{A'PQ} \label{twistor_prop1}\\ \square \kappa_A
   & = \frac{2}{3}\nabla^{PP'}H_{P'PA}
 \end{align}
 where $\square:=g^{ab}\nabla_a\nabla_b$.  Assume that the
 \emph{twistor-candidate equation}
 \begin{equation}\label{TwistorCandidate}
   \square \kappa_A =0.
 \end{equation}
 holds.   Then equation \eqref{twistor_prop1} reduces to
 \begin{equation}
   \square \tensor{H}{_{A'AB}} = 2\tensor{\Psi}{_{AB}^{PQ}}
   \tensor{H}{_{A'PQ}}
 \end{equation} so that if one
 provides trivial initial data for $H_{A'AB}$ on a Cauchy hypersurface
 $\Sigma_0$:
 \begin{equation}\label{Twistor_ID_Raw}
 H_{A'AB}|_{\Sigma_0}=0, \qquad \nabla_{EE'}H_{A'AB}|_{\Sigma_0}=0,
 \end{equation}
 then, by local existence and uniqueness of symmetric hyperbolic
 systems, one has that $H_{A'AB}=0$ on a spacetime neighbourhood
 $\mathcal{U} \subset \mathcal{D}^{+}(\Sigma_0)$, where $D^{+}(\Sigma_0)$
 denotes the future domain of dependence of $\Sigma_0$.  The initial
 conditions \eqref{Twistor_ID_Raw} can be translated in terms of
 $\kappa_A$ as:
 \begin{align}\label{Twistor_ID}
   \nabla_{A'(A}\kappa_{B)}|_{\Sigma_0}=0, \qquad
   \nabla_{EE'}\nabla_{A'(A}\kappa_{B)}|_{\Sigma_0}=0,
  \end{align}
 and are regarded as initial data constraints for equation
 \eqref{TwistorCandidate}. To obtain conditions intrinsic to $\Sigma_0$
 one needs to perform a 1+3 spinor split.  Although this is analogous
 to the standard 3+1 split in tensors, in general, the spacespinor
 split is not adapted to a foliation but rather to a congruence of
 timelike curves with tangent $\tau^{AA'}$ normalised so that
 $\tau^{AA'}\tau_{AA'}=2$. The Levi-Civita covariant derivative
 of any spinor $\mu_{C}$
splits as:
 \begin{equation}
\nabla_{AA'}\mu_{C} = \tfrac{1}{2} \tau _{AA'} \nabla_{\tau }\mu_{C} -
\tau ^{B}{}_{A'} \mathcal{D} _{BA}\mu _{C}
 \end{equation}
 where $\nabla_{\tau}:=\tau^{AA'}\nabla_{AA'}$ and $\mathcal{D}_{AB}:=
 \tau_{(A}{}^{A'}\nabla_{B)A'}$ is the Sen connection relative to
 $\tau^a$ ---see \cite{CFEBook}.  The spacetime covariant derivative of
 $\tau^a$ is determined in terms of the acceleration $\chi_{AB}$ and
 Weingarten spinors $\chi_{ABCD}$ through:
 \begin{equation}
\nabla_{AA'}\tau _{CC'} = - \tfrac{1}{\sqrt{2}} \chi _{CD} \tau _{AA'} \tau
^{D}{}_{C'} + \sqrt{2} \chi _{ABCD} \tau ^{B}{}_{A'} \tau ^{D}{}_{C'}
 \end{equation}
 If $\tau^{AA'}$ is hypersurface orthogonal then
 $\chi_{ABCD}=\chi_{AB(CD)}$ and corresponds to the second fundamental
 form of a foliation $\Sigma_{\tau}$. In addition, the 3-dimensional
 Levi-Civita connection on $\Sigma_{\tau}$ is given by,
 \begin{equation}
   D_{AB} \mu_{C} = \mathcal{D}_{AB}\mu_{C} +
   \frac{1}{\sqrt{2}}\chi_{(AB)C}{}^{Q}\mu_{Q}
 \end{equation}
 Using the spacespinor formalism, in \cite{BaeVal11a}, it was shown that the
 equations \eqref{Twistor_ID} are reduced to the following
 conditions:
 \begin{subequations}
 \begin{eqnarray}
   &\mathcal{D} _{(AB}\kappa _{C)} =0, \label{spatialtwistor} \\ &
   \Psi_{ABCD}\kappa^A = 0 \label{buchdalID} \\ & \nabla_\tau \kappa_A
   = -\frac{2}{3}\mathcal{D}_{A}{}^{B}\kappa_B \label{prescriptionTimeDer}.
 \end{eqnarray}
 \end{subequations}
 
 The conditions \eqref{spatialtwistor} and \eqref{buchdalID} are
 called the \emph{twistor initial data equations}. Notice that
 $\Psi_{ABCD}$ on $\Sigma_0$ can be written in terms of its electric and
 magnetic part respect to $\tau^a$ which in turn can be expressed in
 terms of the initial data set $(\Sigma_0, \bmh, \bm\chi)$ where $\bmh$
 and $\bm\chi$ are the first and second fundamental forms of $\Sigma_0$
 ---see \cite{BaeVal11a} for further details and \cite{GarVal08c} for
 an alternative way of expressing these conditions.  If the twistor
 initial data equations are solved for $\kappa_A$ on $\Sigma_0$ and its
 time derivative on $\Sigma_0$ is prescribed according to equation
 \eqref{prescriptionTimeDer} one obtains the initial data for the
 twistor-candidate equation \eqref{TwistorCandidate} that ensures that
 $\kappa^A$ is an actual twistor in $\mathcal{U}\subset
 \mathcal{D}^{+}(\Sigma_0)$.

 This discussion is summarised in the following propositions.
 
 \begin{proposition}[Garc\'ia-Parrado \& Valiente-Kroon]\label{prop:TwistorInitialData_raw}
    If a spinor $\kappa^A$ satisfies the conditions \eqref{Twistor_ID}
     and solves the vacuum twistor candidate wave equation
    \eqref{TwistorCandidate} then $\kappa^A$ is twistor in a Ricci-flat open set
    $\;\mathcal{U}\subset \mathcal{D}^{+}(\Sigma_0)$.
  \end{proposition}

  \begin{proposition}[B\"ackdahl \& Valiente-Kroon]\label{prop:TwistorInitialData_refined}
    Let $(\Sigma_0,\bmh,\bm\chi)$ be an initial data set for the vacuum
    Einstein field equations (without cosmological constant) where
    $\bmh$ is the 3-metric on a spacelike Cauchy hypersurface $\Sigma_0$
    and $\bm\chi$ is the second fundamental form.  If there exist a non-trivial
    spinor $\kappa^A_{*}$ satisfying the twistor initial data
    equations \eqref{spatialtwistor}-\eqref{buchdalID} on $\Sigma_0$,
    then the spacetime development of $(\Sigma_0,\bmh,\bm\chi)$ will
    posses a twistor $\kappa_A$ in an open set $\mathcal{U}\subset
    \mathcal{D}^{+}(\Sigma_0)$.  The twistor $\kappa^A$ is obtained by
    solving the twistor candidate equation \eqref{TwistorCandidate}
    with initial data $(\kappa^A_{*}, \nabla_\tau \kappa^A_{*})$ on
    $\Sigma_0$ prescribed according to equations
    \eqref{spatialtwistor}-\eqref{prescriptionTimeDer}.
  \end{proposition}

  Generally, initial data for the Einstein field equations satisfying
  conditions \eqref{spatialtwistor}-\eqref{prescriptionTimeDer} is not
  sufficient to ensure that the development will be a pp-wave
  spacetime, as Proposition \ref{prop:TwistorInitialData_refined}
  does not guarantee that $\nabla_{A'Q}\kappa^Q=0$.  To see whether
  imposing further conditions conditions on the initial data is enough
  so that $\nabla_{A'Q}\kappa^Q=0$ on $\mathcal{U}\subset
  \mathcal{D}^{+}(\Sigma_0)$, it suffices to construct a propagation
  equation for the quantity $\bar{\eta}_{A'}:=\nabla_{A'}^Q\kappa_Q$.
  Commuting covariant derivatives and using the spinorial-Ricci
  identities, a calculation similar to that leading to equation
  \eqref{twistor_prop1} gives
  \begin{align*}\label{propEtaIdentity}
    \square \bar{\eta}_{A'} = -2 \Lambda \bar{\eta}_{A'} +
    8\kappa^A\nabla_{AA'}\Lambda.
  \end{align*}
  Thus, if $\Lambda=0$,
   \begin{equation}\label{curlfree_prop_equation}
    \square \bar{\eta}_{A'} = 0.
  \end{equation}
  Hence, by providing trivial initial data for equation
  \eqref{curlfree_prop_equation}:
  \begin{align}\label{Curl_free_ID_Conds}
\bar{\eta}_{A'}|_{\Sigma_0}=0, \qquad
\nabla_{EE'}\bar{\eta}_{A'}|_{\Sigma_0}=0,
  \end{align}
  we get $\bar{\eta}_{A'}=0$ on $\mathcal{U}\subset
  \mathcal{D}^{+}(\Sigma_0)$.  Thus, by augmenting the requirements on
  the initial data implied by equation \eqref{Twistor_ID_Raw} to
  include those encoded in equation \eqref{Curl_free_ID_Conds}, one
  ensures that $\nabla_{A'Q}\kappa^Q=0$ on $\mathcal{U} \subset
  \mathcal{D}^{+}(\Sigma_0)$.  To translate the latter to intrinsic
  conditions on $\Sigma_0$ one needs to perform a 1+3 spacespinor split
  as follows. Let $\zeta _{A} = \bar{\eta} _{A'} \tau
  _{A}{}^{A'}$. Equivalently, $\bar{\eta}_{B'}= - \zeta _{A} \tau
  ^{A}{}_{B'}$.  Then the conditions \eqref{Curl_free_ID_Conds}, are
  equivalent to impose $\zeta_A=0$ and $\nabla_\tau \zeta_A=0$.  A
  direct calculation shows that
  \begin{align}
   \zeta _{A} = - \tfrac{1}{2} \nabla_{\tau }\kappa _{A} + \mathcal{D}
   _{A}{}^{B}\kappa _{B}.
  \end{align}
  Hence using the condition \eqref{prescriptionTimeDer} one gets
   \begin{align}\label{eq:zetaToDerkappa}
  \zeta_A = \frac{4}{3}\mathcal{D} _{A}{}^{B}\kappa_{B}.
  \end{align}
   Similarly, a direct calculation shows that
   \begin{equation}
     \nabla_{\tau }\bar{\eta} _{A'} = \tfrac{4}{3} \tau ^{A}{}_{A'}
     \nabla_{\tau }\mathcal{D} _{AB}\kappa ^{B} - \tfrac{4 \sqrt{2}}{3} \chi
     ^{A}{}_{C} \tau ^{C}{}_{A'} \mathcal{D} _{AB}\kappa ^{B}.
   \end{equation}
   Now, recall that the commutator $[\nabla_{\tau },\mathcal{D} _{AB}]$
   acting on any spinor $\mu_C$
   reads
   \begin{align}
     [\nabla_{\tau },\mathcal{D} _{AB}]\mu_{C} & = \Psi _{ABCD} \mu^{D}   - 2\Lambda  \mu_{(A} \epsilon _{B)C}  - \Phi _{CDA'B'} \mu^{D} \tau _{A}{}^{A'} \tau _{B}{}^{B'}
     - \tfrac{1}{\sqrt{2}}\chi _{AB} \nabla_{\tau  }\mu_{C} \nonumber \\ & + \tfrac{2}{\sqrt{2}}\chi _{(A}{}^{D} \mathcal{D} _{B)D}\mu_{C}
     - \sqrt{2} \chi _{(AB)DF} \mathcal{D} ^{DF}\mu_{C}
   \end{align}
   Taking $\mu_{C}=\kappa_{C}$
   and using equations \eqref{prescriptionTimeDer} and \eqref{eq:zetaToDerkappa},
   a long but straightforward
   calculation renders   
   \begin{align}
   	\nabla_{\tau  }\zeta _{A} & = - \sqrt{2} \chi _{AB} \zeta ^{B} +
   	4 \Lambda  \kappa _{A} - \tfrac{2}{3} \sqrt{2} \zeta ^{B} 
   	\chi _{A}{}^{C}{}_{BC} - \tfrac{4}{3} \Phi _{BCA'B'} \kappa^{B} \tau _{A}{}^{A'} \tau ^{CB'} \nonumber
        \\ & + \tfrac{2}{3} \mathcal{D}_{AB}\zeta^{B} + \tfrac{2 \sqrt{2}}{3} \chi ^{BC}
        \mathcal{D}_{(AB}\kappa _{C)} - \tfrac{4}{3} \sqrt{2} \chi _{A}{}^{BCD} \mathcal{D} _{(BC}\kappa _{D)}
   \end{align}
   Assuming vacuum and using condition \eqref{spatialtwistor},
   simplifies the latter expression to
     \begin{align}
     \nabla_{\tau }\zeta _{A} & = - \sqrt{2}\chi _{AB} \zeta ^{B} -
     \tfrac{2}{3} \sqrt{2} \zeta ^{B} \chi _{A}{}^{C}{}_{BC} +
     \tfrac{2}{3} \mathcal{D} _{AB}\zeta ^{B} .
     \end{align}
     Then, with these assumptions, the conditions $\zeta_A=\nabla_\tau
     \zeta_A=0$ reduce to the requirement that
     $\mathcal{D} _{A}{}^{B}\kappa_{B}=0$.  Altogether the condition
     $\mathcal{D} _{(AB}\kappa _{C)} =0$ and $\mathcal{D}
     _{A}{}^{B}\kappa_{B}=0$ can be encoded simply as $D_{AB}\kappa_{C}=0$.

     \begin{remark}
       \emph{
     Observe that $\kappa_A=0$ trivially solves the condition $D_{AB}\kappa_{C}=0$.
     Hence, in the sequel we will be concerned only with non-trivial initial data,
     namely a spinor $\kappa_{*A} \neq 0$ everywhere on $\Sigma_0$ that satisfies
     $D_{AB}\kappa_{*C}=0$.     }
     \end{remark}
     
     This discussion is summarised in the following:
    \begin{proposition}\label{prop:Twistor_PP_wave_InitialData}
    Let $(\Sigma_0,\bmh,\bm\chi)$ be an initial data set for the
    vacuum Einstein field equations (without cosmological constant)
    where $\bmh$ is the 3-metric on a Cauchy hypersurface $\Sigma_0$
    and $\bm\chi$ is the second fundamental form.  If the conditions
    \vspace{-.75cm}\begin{center}
    	\begin{equation}
    		\mathcal{D} _{AB}\kappa _{C}
                =0 \label{TwistorAugmented_ID_1}
    	\end{equation}\vspace{-5mm}
    	\begin{equation}
    		\Psi_{ABCD}\kappa^A = 0. \label{TwistorAugmented_ID_2}
    	\end{equation}
    \end{center}
  are satisfied by a non-trivial spinor $\kappa_{*}^A$ on $\Sigma_0$,
  then a covariantly constant spinor $\kappa^A$ in some open set
  $\mathcal{U}\subset \mathcal{D}^{+}(\Sigma_0)$ is obtained by
  solving the twistor candidate equation \eqref{TwistorCandidate} with
  initial data $(\kappa^A_{*}, \nabla_\tau \kappa^A_{*})$ on
  $\Sigma_0$ prescribed according to equation
  \eqref{prescriptionTimeDer}.
    \end{proposition}

    
    \begin{remark}[Continuity argument] \label{Remark_continuity}\emph{
      Notice that proposition \ref{prop:Twistor_PP_wave_InitialData}
      does not exclude the possibility that the covariantly constant
      spinor $\kappa^A$ becomes trivial ($\kappa^A=0$) at some point
      in the evolution. In other words, although by assumption the
      initial data for the wave equation \eqref{TwistorCandidate} is
      non-trivial, $\kappa_{*}^{A} \neq 0$, this condition alone does
      not guarantee that $\kappa^A \neq 0$ in the whole domain of
      dependence $\mathcal{D}^{+}(\Sigma_0)$.  However, if the
      solution $\kappa_A$ is a classical solution ($C^2$) of the
      wave equation \eqref{TwistorCandidate}, then \emph{by
      continuity}, it follows that in a small spacetime
      neighbourhood of the initial hypersurface $\Sigma_0$ the
      solution is non-trivial: $\kappa_A \neq 0$ in $\mathcal{U}\subset \mathcal{D}^{+}(\Sigma_0)$.
      Therefore, although
      we cannot control the size of the spacetime neighbourhood
      $\mathcal{U}$ where the solution is non-trivial, it is clear by
      continuity that, by shrinking the size of the set
      $\mathcal{U}\subset \mathcal{D}^{+}(\Sigma_0)$, proposition
      \eqref{prop:Twistor_PP_wave_InitialData} implies, in conjuction
      with proposition \ref{prop1_ppwave}, that the development of
      $(\Sigma_0,\bmh,\bm\chi)$ will be a pp-wave spacetime on a small
      spacetime neighbourhood $\mathcal{U}$ close to the initial
      hypersurface.}
    \end{remark}

    The requirement on the regularity of $\kappa_A$
    (needed for the continuity argument of Remark
    \ref{Remark_continuity}) can be transformed to a condition at
    the level of initial data by employing
    a standard existence and uniqueness theorem for wave equations.
    Although there could be other options, in the following we will
    make use of the local existence and uniqueness result for wave
    equations of \cite{HugKatMar77} with initial data
    in some suitable Sobolev space $H^m$.  We use this result in the form
    presented in Theorem 2 in Appendix E of \cite{GasVal15}.

    \begin{remark}[Regularity of initial data]\label{ID-regularity_remark}
       \emph{ Introduce some local coordinates $x=(x^\mu)=(\tau,x^i)$
       in $\mathcal{U}$, with $\mu=0,1,2,3$ and $i=1,2,3$ so that
       $\Sigma_0$ is described by $\tau=0$. Denote the components of
       $\kappa_A$ as $\bm\kappa$. Then, the wave equation
       \eqref{TwistorCandidate} written in local coordinates reads:
      \begin{equation*}
        g^{\mu\nu}\partial_\mu\partial_\nu \bm\kappa =F(x,\bm\kappa,
        \partial\bm\kappa),
      \end{equation*}
      where $F$ is linear in $\bm\kappa$ and $\partial
      \bm\kappa$ and $g_{\mu\nu}$ is a Lorentzian metric.  We will
      consider non-trivial solutions $\bm\kappa_{*}$ to the initial
      data equations \eqref{TwistorAugmented_ID_1} and
      \eqref{TwistorAugmented_ID_2} which satisfy the following
      regularity conditions.  For $m \geq 4$ :
      \begin{subequations}\label{regcond}
      \begin{flalign}
      & \bm\kappa_{*}\in H^m(\Sigma_0, \mathbb{C}^2) \quad \text{and}
        \quad \partial_\tau\bm\kappa_{*}\in H^m(\Sigma_0, \mathbb{C}^2),
        \\ & (\bm\kappa_{*},\; \partial_\tau\bm\kappa_{*}) \in
        D_{\delta} \qquad \text {for some $\delta>0$ where} \nonumber
        \\ & D_\delta \equiv \big\{ (w_1,w_2)\in
        H^m(\Sigma_0,\mathbb{C}^2)\times H^m(\Sigma_0,\mathbb{C}^2) \;
        | \; \delta < | \det g_{\mu\nu} | \big\}.
      \end{flalign}
      \end{subequations}
      With these assumptions, then using
      point (i) of Theorem 2 of \cite{GasVal15} one concludes that
      there exists $T>0$ and a unique solution to the Cauchy problem
      defined on $[0,T)\times \Sigma_0$ such that
    \[\bm \kappa \in C^{m-2}([0,T)\times \Sigma_0, \mathbb{C}^2).  \]
      }
 \end{remark}

    Combining propositions \ref{prop:Twistor_PP_wave_InitialData} 
    and \ref{prop1_ppwave}, and aided with the discussion of Remarks \ref{Remark_continuity} and
    \ref{ID-regularity_remark}
    one obtains the following:
    
    \begin{theorem}\label{prop:Twistor_PP_wave_InitialData_refined}
    Let $(\Sigma_0,\bmh,\bm\chi)$ be an initial data set for the
    vacuum Einstein field equations (without cosmological constant)
    where $\bmh$ is the 3-metric on a Cauchy hypersurface $\Sigma_0$
    and $\bm\chi$ is the second fundamental form.  If the conditions
    \vspace{-.75cm}\begin{center}
    	\begin{equation}
    		\mathcal{D} _{AB}\kappa _{C}
                =0 \label{TwistorAugmented_ID_1}
    	\end{equation}\vspace{-5mm}
    	\begin{equation}
    		\Psi_{ABCD}\kappa^A = 0. \label{TwistorAugmented_ID_2}
    	\end{equation}
    \end{center}
    are satisfied by a non-trivial spinor $\kappa_{*}^A$ on
    $\Sigma_0$ satisfying the regularity conditions of Remark \ref{ID-regularity_remark}.
    Then there exist a open set $\mathcal{U}\subset
    \mathcal{D}^{+}(\Sigma_0)$ ---a possibly very small spacetime
    neighbourhood of the initial hypersurface $\Sigma_0$--- for which
    the development of $(\Sigma_0,\bmh,\bm\chi)$ on $\mathcal{U}$ is
    a pp-wave spacetime for some function $\mathcal{H}$ satisfying
    equation \eqref{HppwaveHarmonic} in $\mathcal{U}$.
    \end{theorem}
    
 \subsection*{Acknowledgements}
E. Gasper\'in holds a FCT (Portugal) investigator grant
2020.03845.CEECIND. We have profited from scientific discussions with
Filipe C. Mena.  This project was supported by the H2020-MSCA-2022-SE
project EinsteinWaves, GA No. 101131233. E. Gasper\'in has also benefited
from the RISE project at CENTRA through the European Union’s H2020 ERC
Advanced Grant “Black holes: gravitational engines of discovery” grant
agreement no. Gravitas–101052587, also the Villum Foundation (grant
no. VIL37766) and the DNRF Chair program (grant no. DNRF162) by the
Danish National Research Foundation and the Marie Sklodowska-Curie
grant agreement No 101007855 and No 101131233.
We are grateful for the comments of the anonymous referees on this note.

\normalem
\bibliographystyle{unsrt}


\end{document}